\documentclass[10pt]{exam}

\usepackage[utf8]{inputenc}
\usepackage{amsmath, amssymb, graphicx, caption, float, wasysym, tcolorbox, tikz, physics, setspace, blindtext, csquotes, dsfont, amsthm, newpxtext, mathtools}

\DeclareMathOperator*{\argmax}{arg\,max}

\usepackage[symbol]{footmisc}
\usepackage[colorlinks=true, allcolors=blue]{hyperref}

\newtheorem{theorem}{Theorem}[section]
\newtheorem{definition}{Definition}[section]
\newtheorem{example}{Algorithm}[section]

\newtheorem{lemma}[theorem]{Lemma}

\makeatletter
\def\@xfootnote[#1]{%
  \protected@xdef\@thefnmark{#1}%
  \@footnotemark\@footnotetext}
\makeatother

\begin{document}

\begin{titlepage}
    \begin{center}
        \vspace*{1cm}
            
        \LARGE
        \textsc{Prophet inequalities for subadditive combinatorial auctions}
            
        \vspace{1cm}
        \large
        \text{Dwaipayan Saha}\footnote[$\dagger$]{Department of Computer Science, Princeton University, 35 Olden Street, Princeton 08540, US, email:
\texttt{dsaha@princeton.edu}} \quad \text{Ananya Parashar}\footnote[$\ddagger$]{Department of Operations Research and Financial Engineering, Princeton University, 98 Charleston Street, Princeton 08540, US, email:
\texttt{parashar@princeton.edu}}
        \vspace{0.5cm}
        \begin{abstract}
    In this paper, we survey literature on prophet inequalities for subadditive combinatorial auctions. We give an overview of the previous best $O(\log \log m)$ prophet inequality as well as the preceding $O(\log m)$ prophet inequality. Then, we provide the constructive posted price mechanisms used in order to prove the two bounds. We mainly focus on the most recent literature that resolves a central open problem in this area of discovering a constant factor prophet inequality for subadditive valuations. We detail the approach of this new paper, which is \textit{non-constructive} and therefore cannot be implemented using prices, as done in previous literature.

        \end{abstract}
            
        \vspace{1cm}

        \vspace{1cm}

    \end{center}
    
\end{titlepage}

\begin{spacing}{1.3}

\tableofcontents
\newpage
\section{Introduction}

Combinatorial auctions have been a topic of high interest for countless years now. This setting is one where items are divided amongst bidders, and the auctioneer doesn't know individual bidders' preferences for each item prior to their arrival. In such a setting, we hope to maximize welfare, or the utility gained by each bidder for the items they receive. Prophet inequalities help us do so by providing a bound on the maximum possible welfare that can be achieved in these auctions.

While not immediately obvious, modeling mechanisms for the combinatorial auction problem can be extended to a wide range of real-world settings. Wireless networks is a major real-world application as different frequencies must be allocated in order to provide different types of communication services. We may also model advertising spaces as a combinatorial auction as in the world of online advertising, different combinations of advertising slots on websites and other platforms can provide different levels of exposure to potential customers. These combinations can be modeled as auctions, and prophet inequalities can be used to design mechanisms for these auctions that will find near-optimal allocations of advertising slots. We specifically look at subadditive auctions in our paper as we see in the above examples that items may be valued differently if allocated individually than allocated together.




In this paper, we survey various literature on prophet inequalities for subadditive combinatorial auctions. We focus mainly on the recent work by José Correa and Andrés Cristi in proving the existence of a $(6 + \varepsilon)$ prophet inequality for subadditive valuations \cite{constant}, achieving a major breakthrough from the previous paper by Dütting, Lucier, and Kesselheim which found an $O(\log \log m)$ prophet inequality for subadditive valuations \cite{loglog}. We will also introduce other significant previous bounds: an $O(1)$ prophet inequality for a larger class of valuations, XOS, and an $O(\log m)$ prophet inequality for subadditive valuations by Dütting, Feldman, Lucier, and Kesselheim \cite{log}.

It is important to note that the work by Correa and Cristi provides a \textit{nonconstructive} existence proof of the constant factor prophet inequality, rather than algorithms which can be implemented like the preceding $O(\log \log m)$ and $O(\log m)$ approaches discussed above. Both of these previous works use posted price mechanisms which we discuss later. While this work cannot be implemented yet, the unique tools and ideas presented in this paper in conjunction with other pricing mechanisms can help develop a constructive constant-time mechanism.

\section{Preliminaries}

\subsection{Model}

We first describe this problem informally. Consider the setting of an auction with $m$ items and $n$ bidders. Each bidder assigns a value to every subset of $m$ items with valuation function $v_i$ for $i \in [n]$. This is the simple auction that most are familiar with, but now rather than all people bidding at the same time, we consider the case where the $n$ bidders arrive sequentially. While this order may be arbitrary, assume without loss of generality that they arrive in the order $1, 2, \dots , n$ \footnote[1]{This can be formalized by considering a $\sigma: [n] \rightarrow N$, where $\sigma$ maps the ordering to the set $N$ of people and the analysis still holds.}. The auctioneer has knowledge of the distribution of the value functions, but not the individual value functions themselves prior to auction. As each person $i \in [n]$ arrives, their valuation function $v_i$ is revealed to the auctioneer upon which a subset of available items is awarded to person $i$ by our algorithm. Doing so for every person determines the allocation realized by our algorithm with respect to the valuations revealed throughout the process.

Formally, we have a set $N$ of $n$ people and a set $M$ of $m$ items. We notate $X_i$ as the outcome space, or set of allocations to each buyer $i \in N$. Moreover, $\varnothing \in X_i$ for all $i \in N$, as we know each buyer can simply receive nothing. Now the joint outcome space, or the space of all possible allocations, can be denoted by $X = X_1 \times \dots \times X_n$. We define $\vb{x}_S \in X$ as an allocation only to those buyers in the subset $S \subseteq N$:

$$
\vb{x}_S  = \begin{cases}
  x_i  & i \in S \\
  \varnothing & i \notin S
\end{cases}
$$

Moreover, $\vb{x}_{[i-1]}$ is defined similarly where $S = \{1, 2, \dots, i-1\}$. That is, we only allocate items to buyers $1$ through $i-1$, and buyers $i$ through $n$ receive $\varnothing$. While these are all indeed allocations, we denote $\mathcal{F} \subseteq X$ as only those that are feasible. We require that if $\vb{x}$ is feasible then $\vb{x}_S$ is feasible (or, more generally, $\mathcal{F}$ is downward closed).

Each buyer $i$ has a valuation function $v_i: 2^M \rightarrow \mathbb{R}_{\geq 0}$, or similarly $v_i: X_i \rightarrow \mathbb{R}_{\geq 0}$ since either $2^M$ or $X_i$ notate the set of possible allocations to buyer $i$. We assume these valuations are both normalized and monotone, which means $v_i(\varnothing) = 0$ and for $S \subseteq T \subseteq M$, we have $v_i(S) \leq v_i(T)$. Lastly, we assume that the values are finite, i.e. for all $i \in N$ and $A \subseteq M$, we have $\mathbb{E}[v_i(A)] < \infty$.

Futhermore, every $v_i$ is drawn independently from some distribution $\mathcal{D}_i$. We notate $\mathcal{D} = \times_i \mathcal{D}_i$ as the product or joint distribution. Here, $\mathcal{D}_i$ is a distribution over the set of valuations $V_i$ for person $i \in N$. Again, assume that the auctioneer has knowledge of every distribution but not the values themselves prior to auction. 


    \begin{definition}
        The classes of valuation functions in increasing order of generality are as follows:
        \begin{itemize}
            \item \textbf{Submodular:} A valuation $v$ is submodular if for any $S, T \subseteq M$ such that $S \subseteq T$ we have that $v(S \cup \{j\}) - v(S) \geq v(T \cup \{j\}) - v(T)$. Such functions can be interpreted as representations of diminishing marginal utility.
            \item \textbf{XOS:} A valuation $v$ is XOS if there are a collection of additive supports $a_1, \dots, a_k$ such that for any $S \subseteq M$ we have $v(S) = \underset{1 \leq i \leq k}{\max}a_i(S)$.

            An equivalent definition says that a valuation $v$ is XOS if for any $S \subseteq M$, there exists a fractional covering $\{\lambda_i, T_i\}_{i = 1}^k$ of $S$ such that $v(S) \leq \sum_{i = 1}^k \lambda_i v(T_i)$. A fractional covering of $S$ means for $\lambda_i > 0, T_i \subseteq M$ and for all $j \in S$ we have $\sum_{i:j\in T_i}\lambda_i \geq 1$\footnote[2]{Showing the two definitions to be equivalent requires use of duality and is left as exercise for the reader.}.

            \item \textbf{Subadditive:} A valuation $v$ is subadditive if for any $S, T \subseteq M$ we have that $v(S \cup T) \leq v(S) + v(T)$. This is the idea of items not necessarily being more valuable together than apart.
        \end{itemize}
    \end{definition}
    
We notate $\vb{v} = (v_1, \dots, v_n)$ as the vector of specific valuations, or a valuation profile. Utility can be defined as $u_i = v_i(x_i) -\pi_i$ for buyer $i$. This is the price $\pi_i$ buyer $i$ pays for the items subtracted from the value for items received $v_i(x_i)$. The welfare is defined to be $\vb{v}(\vb{x}) = \sum_{i}v_i(x_i)$ for some allocation $\vb{x}$.

Notate $\textsf{OPT}(\vb{v}) = (U_1, \dots, U_n)$ as the optimal allocation, where the optimal allocation has the maximum welfare over all feasible allocations $\textsf{OPT}(\vb{v}, \mathcal{F}) = \argmax_{\vb{x} \in \mathcal{F}} \vb{v}(\vb{x})$. $\vb{v}(\textsf{OPT}(\vb{v}))$ is the total value the allocation achieves. Similarly, we write $\textsf{ALG}(\vb{v}) = (x_1, \dots, x_n)$ as the algorithm's allocation and $\vb{v}(\textsf{ALG}(\vb{v}))$ as the value it achieves. $\textsf{ALG}_i(\vb{v})$ is the algorithm's allocation for buyer $i$ for the valuation profile $\vb{v}$.

\subsection{Prophet inequalities}

Prophet inequalities allow us to compare the worst case competitive stochastic ratio of our algorithm or mechanism to the optimal allocation in hindsight:
$$\tau = \sup_{\mathcal{D}} \frac{\mathbb{E}_{\vb{v} \sim \mathcal{D}}[\vb{v}(\textsf{OPT}(\vb{v}))]}{\mathbb{E}_{\vb{v} \sim \mathcal{D}}[\vb{v}(\textsf{ALG}(\vb{v}))]}$$
It is evident that an equivalent form is:
$$\tau \cdot \mathbb{E}_{\vb{v} \sim \mathcal{D}}[\vb{v}(\textsf{ALG}(\vb{v}))] \geq \mathbb{E}_{\vb{v} \sim \mathcal{D}}[\vb{v}(\textsf{OPT}(\vb{v}))]$$
\noindent which would provide a $\tau$ approximation or a $\tau$ prophet inequality. In this case, $\tau$ could also be referred to as the competitive ratio of our algorithm. In general, these are tools that allow one to conveniently lower bound the relative performance of an algorithm's performance to the optimum.

                


\subsection{Posted Price Mechanisms}

Posted price mechanisms define pricing rules $p_i: 2^M  \rightarrow \mathbb{R}_{\geq 0}$, which assign a non-negative price to each subset of items $S \subseteq M$. Now we define three types of prices:

\begin{itemize}
    \item \textit{Static prices} are prices independent of the items that have already been sold. That is, for any $S \subseteq M$ we have $p_i(S) = p_i(S|\vb{x}_{[i - 1]})$. 

\item \textit{Anonymous Prices} for an item are simply those that do not depend on the person purchasing it. In other words, for any $i, j \in N$ and $S \subseteq M$ we have $p_i(S) = p_j(S)$.

\item \textit{Item prices} mean that the price of a set of items is simply the sum of the prices of the items within it, or $p_i(S) = \sum_{j \in S} p_j$.
\end{itemize}
Note that there are various types of posted price mechanisms. In this paper, we begin by exploring balanced prices from \cite{log} and then prices as defined in \cite{loglog}. Furthermore, we only consider \textit{static, anonymous, and item} prices. It has been shown that balanced prices generalize to dynamic pricing rules as well in \cite{log}.

Such prices can be used to get good approximations of the optimal welfare by considering welfare to be comprised of utility and revenue. We bound each in order to evaluate the final revenue. This technique will be illustrated throughout later sections.

\section{Related Work}

\subsection{$O(\log m)$ Balanced Price Mechanism \cite{log}}

\subsubsection{Bayesian Setting with $1$ item}

To provide intuition for the power of this mechanism, we consider the single item case, where each person has a value in $\mathbb{R}_{\geq 0}$ for the item. Furthermore, for any bidder $i \in N$, their value $v_i$ is drawn from $\mathcal{D}_i$ independently. The joint distribution is $\mathcal{D} = \times_i\mathcal{D}_i$ and is known by the auctioneer ahead of time. Once again, people arrive sequentially.

Now, we post a balanced price of $p = \frac{1}{2}\mathbb{E}_{\vb{v} \sim \mathcal{D}}[\max_i v_i]$ and notate $q$ as the probability of it being sold. In order to approximate the welfare, we use a combination of a revenue and utility bound. First, note that the expected revenue is simply $\frac{1}{2}\mathbb{E}_{\vb{v} \sim \mathcal{D}}[\max_i v_i] \cdot q$. This is because, if the item sold with probability $q$, we receive $\frac{1}{2}\mathbb{E}_{\vb{v} \sim \mathcal{D}}[\max_i v_i]$ and with probability $1-q$, our revenue is $0$. Notate the indicator random variable $X_i := \mathds{1}(\text{item was still available when person $i$ arrived})$. Notate the event that the item was never sold to be $\{X_{n+1} = 1\}$. Then consider the utility bound:
\begin{align*}
    \mathbb{E}[u(\vb{v})] &= \sum_{i\in N} \mathbb{E}[(v_i - \frac{1}{2}\mathbb{E}_{\vb{v} \sim \mathcal{D}}[\max_i v_i])^+ \cdot X_i]\\
    &= \sum_{i\in N} \mathbb{E}[(v_i - \frac{1}{2}\mathbb{E}_{\vb{v} \sim \mathcal{D}}[\max_i v_i])^+] \cdot \mathbb{E}[X_i] \quad \text{by independence}\\
    &= \sum_{i\in N} \mathbb{E}[(v_i - \frac{1}{2}\mathbb{E}_{\vb{v} \sim \mathcal{D}}[\max_i v_i])^+] \cdot \Pr[X_i = 1] \geq \sum_{i \in N} \mathbb{E}[(v_i - \frac{1}{2}\mathbb{E}_{\vb{v} \sim \mathcal{D}}[\max_i v_i])^+] \cdot (1-q)
\end{align*}

\noindent The last inequality holds since $1-q = \Pr[X_{n+1} = 1] \leq \Pr[X_i = 1]$ for all $i \in [n]$. Now combining our bounds:
\begin{align*}
    \mathbb{E}[\textsf{ALG}(\vb{v})] &\geq \mathbb{E}[r(\vb{v})] + \mathbb{E}[u(\vb{v})]\\
    &\geq \frac{1}{2}\mathbb{E}_{\vb{v} \sim \mathcal{D}}[\max_i v_i] \cdot q + \sum_{i\in N} \mathbb{E}_{\vb{v} \sim \mathcal{D}}[(v_i - \frac{1}{2}\mathbb{E}_{\vb{v} \sim \mathcal{D}}[\max_i v_i])^+] \cdot (1-q)\\
    &\geq \frac{1}{2}\mathbb{E}_{\vb{v} \sim \mathcal{D}}[\max_i v_i] \cdot q + \mathbb{E}_{\vb{v} \sim \mathcal{D}}[\max_i v_i  - \frac{1}{2}\mathbb{E}_{\vb{v} \sim \mathcal{D}}[\max_i v_i]] \cdot (1-q) = \frac{1}{2}\mathbb{E}_{\vb{v} \sim \mathcal{D}}[\max_i v_i]
\end{align*}
We use that the maximum is lower than the sum, and furthermore the expression is clearly non-negative allowing us to drop $(\cdot)^+$. This gives a $2$-approximation for the expected welfare with this technique of balancing the price. If the item was bought early, the algorithm does well by generating revenue, and generating utility otherwise.

\subsubsection{Extension to Multiple Items}

Reverting to the regular model with multiple items, we first provide intuition of good balanced prices. The general idea of balanced prices is that they are sufficiently low. That is, if the optimal set $U_i$ for person $i$ is available, then they are willing to purchase it: $\sum_{j \in U_i} p_j \leq v_i(U_i)$. Moreover, they must be sufficiently high. If some subset $T \subseteq U_i$ of person $i$'s optimal subset gets taken earlier by some person $j$ then the loss in utility for person $i$ is at most the revenue generated by person $j$. This can be written as: $\sum_{j \in T} p_j \geq [v_i(U_i) - v_i(U_i \setminus T)]$ for any set $T \subseteq M$. This ``balance'' provides intuition as to why setting up such a posted price mechanism would lead to good approximations. Below, we give the general definition of balanced prices and a theorem that proves the guarantee they provide in \cite{log}.


\subsubsection{Formal Global Guarantee}
Before we state the following, we define the term exchange compatible as the set of outcome profiles $\mathcal{H} \subseteq X$ where for a feasible allocation $\vb{x} \in \mathcal{F}$ and for all allocations $\vb{y} \in \mathcal{H}$, we may replace the allocation for any person $i \in N$, $x_i$, with $y_i$ such that the resulting allocation is still feasible. Mathematically, for all $\vb{y} \in \mathcal{H}$ and people $i \in N$ we have that $(y_i, \vb{x}_{-i}) \in \mathcal{F}$. We may extend this to a family of sets $(\mathcal{F}_{\vb{x}})_{\vb{x} \in X}$ which are exchange compatible if set $\mathcal{F}_{\vb{x}}$ is exchange compatible with allocation $\vb{x}$ for all $\vb{x} \in X$. Now, we use the definition and accompanying theorem from \cite{log}, which state the following:

\begin{definition}[Key Balanced Prices Definition]
\label{5}
Let $\alpha> 0, \beta \geq 0$. Given a set of feasible allocations $\mathcal{F}$ and a valuation profile $\vb{v}$, a pricing rule $\vb{p}$ is $(\alpha, \beta)$ balanced with respect to the allocation rule $\textsf{OPT}$ \footnote[3]{In the paper, this definition considers pricing rules to be balanced with respect to an allocation rule \textsf{ALG} since it makes computation of prices easier. Here, we take the specific allocation rule to be the optimal allocation to showcase its relevance to prophet inequalities.}, an exchange-compatible family of sets $(\mathcal{F}_{\vb{x}})_{\vb{x} \in X}$, and indexing of the players $i \in [n]$ if for all $\vb{x} \in \mathcal{F}$:

\begin{enumerate}
    \item $\sum_{i \in N} p_i(x_i) \geq \frac{1}{\alpha} \left(\vb{v}(\textsf{OPT}(\vb{v})) - \vb{v}(\textsf{OPT}(\vb{v}, \mathcal{F}_{\vb{x}}))\right)$
    \item For all $\vb{x}' \in \mathcal{F}_{\vb{x}}, \sum_i p_i(x'_i) \leq \beta (\textsf{OPT}(\vb{v}, \mathcal{F}_{\vb{x}}))$
\end{enumerate}
\end{definition}


\begin{theorem}[Global Guarantee]
\label{6}
Suppose that the collection of pricing rules $(\vb{p^v})_{\vb{v} \in V}$ for feasible outcomes $\mathcal{F}$ and valuation profile $\vb{v} \in V$ is $(\alpha, \beta)$-balanced with respect to the allocation rule $\textsf{OPT}$ and indexing of the players $i \in [n]$. Then for $\delta = \frac{\alpha}{1 + \alpha \beta}$, the posted-price mechanism with pricing rule $\delta \vb{p}$, where $p_i(x_i) = \mathbb{E}_{\widetilde{\vb{v}}}[p_i^{\widetilde{\vb{v}}}(x_i)]$, generates welfare at least $ \frac{1}{1 + \alpha \beta} \cdot \mathbb{E}_{\vb{v}}[\vb{v}(\textsf{OPT}(\vb{v}))]$ when approaching players in the order they are indexed.
\end{theorem}

It was shown that XOS functions admit $(1,1)$-balanced prices and therefore provide a $2$-approximation, which is the tightest current bound for XOS valuations \cite{log}. In the same paper, they also provide an explicit pricing rule which is $(1,1)$-balanced. From this, we obtain an $O(\log m)$ prophet inequality for general subadditive valuations since subadditive valuations can be approximated via XOS valuations with a logarithmic penalty. A formal proof of the fact uses a lemma from \cite{bawal} and can be found at \cite{dwip}.

Three years later, the same authors presented a new posted price technique that bridged the gap from $O(\log m)$ to $O(\log \log m)$, discussed in the next section. However, the balanced price approach is easier to extend than this new approach, because it derives a global guarantee from a pointwise conclusion, whereas the new mechanism does not provide such a translation.

\subsection{$O(\log \log m)$ Posted Price Mechanism \cite{loglog}}

The intuition for these prices stem from exactly what balanced prices fail to capture regarding subadditive valuations. Balanced prices have the unique property that the sum of the prices ``nearly'' approximate the optimal welfare, which leads to most items either being sold or affordable in that prices are not too high.

An example of a subadditive function is $f(x) = \sqrt{x}$, since $f(x+y) \leq f(x) + f(y)$ for all $x, y \geq 0$. The intuition from this example is that a large increase in the input does not correspond to a significant change in the output, due to the function flattening quickly. This explains why, with subadditive valuations, there are instances where it would be beneficial to only consider a smaller subset of the items available for purchase that capture a majority of the optimal welfare. Dütting, Lucier, Kesselheim achieve this by pricing items higher and targeting a specific fraction of items. In fact, they proved that the specific fraction they target always exists.

Below, we provide the key lemma proved and the corresponding theorem it implies. We give the lemma presented in the complete information case and its accompanying theorem.\footnote[4]{The proof of both for the Bayesian Setting is quite similar, but requires further care in lower bounding the value of the game and is provided in Appendix A and B of \cite{loglog} respectively. In the paper, they define new prices to provide a global guarantee instead of a pointwise guarantee, highlighting the deficiency stated at end of Section $3.1$}

\begin{lemma}
\label{subgood}
    For every $i \in N$, subadditive valuation $v \in V_i$, and set $U \subseteq M$ there exists prices $p_j$ for $j \in U$ and a probability distribution $\delta$ over $S \subseteq U$ such that for all $T \subseteq U$:
    $$\sum_{j \in T} p_j + \sum_{S \subseteq U}\delta_S\left(v(S\setminus T) - \sum_{j \in S}p_j \right)\geq \frac{v(U)}{\alpha}$$
    where $\alpha \in O(\log \log m)$.
\end{lemma}

\begin{proof}
    The proof of the above lemma requires formulation of the problem with linear programming, applying strong duality then defining the corresponding $\min$-$\max$ game for the dual and finally lower bounding the value of the game. Thus, we defer the interested reader to \cite{loglog}.
\end{proof}

\begin{theorem}
    For subadditive combinatorial auctions, there is an $O(\log \log m)$-competitive posted price mechanism that uses static anonymous item prices.
\end{theorem}

\begin{proof}
    This proof sketch uses a similar outline as the proof presented in the Balanced Price section, where we lower bound the optimal welfare as the sum of revenue and utility for the auction. In order to do so, we consider the optimal allocation $\textsf{OPT}(\vb{v}) = (U_1, \dotsm U_n)$ with respect to fixed valuation profile $\vb{v}$. Next, for each person $i \in N$ and item $j \in U_i$, consider the prices $p_j$ promised by Lemma \ref{subgood}. Let $\delta$ be the probability distribution over $S \subseteq U_i$ promised by Lemma \ref{subgood}. At this point, we lower bound utility while allowing person $i \in N$ to purchase some $S \subseteq U_i$ paying its respective price, applying Lemma \ref{subgood}, and summing over all bidders. Revenue is just the sum of the prices of items that were sold. Adding the two lower bounds proves the desired lower bound on welfare.
\end{proof}

\section{Main Result from \cite{constant}}

So far, we have seen two constructive posted price mechanisms that achieve $O(\log m)$ and $O(\log \log m)$ competitive ratios respectively. In this section, we will explore the newest result from \cite{constant} which proves the existence of an $O(1)$ competitive ratio in subadditive combinatorial auctions.




\subsection{Random Score Generators}

    \begin{definition}[Random Score Generators]
    A random score generator (RSG) for a person $i \in N$ is a function $\mathcal{S}_i: V_i \rightarrow \Delta(\mathbb{R}_{\geq 0}^M)$, which takes some valuation $v_i \in V_i$ and outputs a distribution $\mathcal{S}_i(v_i)$.
\end{definition}

Such a definition is useful since when valuation $v_i$ is revealed, we can sample values from the distribution: $b_i = \{b_{i,j}\}_{j = 1}^m \sim \mathcal{S}_i(v_i)$ where $\{b_{i,j}\}_{j = 1}^m$ can be interpreted as the per item scores that person $i$'s valuation function assigns each item $j \in M$ and $b_i$ is a vectorized representation in $\mathbb{R}^m$.

    \begin{definition}[Instance of Random Score Generators] We define an Instance of Random Score Generators (IRSG) $\{\mathcal{S}_i\}_{i = 1}^n$ to be an arbitrary collection of score generators for each person $i \in N$.
\end{definition}

We will use such instances for the algorithm presented as well as its analysis. In fact, they provide a constructive $24$-approximation for XOS valuations \cite{constant}, looser than the current tightest bound.

\subsection{Algorithm}

The key algorithm we present in this survey takes advantage of sampling from the given distributions to create imaginary valuations and scores, using them as thresholds of sorts. It's a neat trick and used across literature in this domain. Before getting into the algorithm, we first define for all $i \in N$:
$$R_i := M \setminus \cup_{j = 1}^{i - 1}x_j$$
In other words, $R_i$ represents the set of items still available when person $i$ arrives, which is essential to generate a feasible allocation. Given the product distribution $\mathcal{D}$ and IRSG $\{\mathcal{S}_i\}_{i = 1}^n$, we consider the following algorithm:

\begin{tcolorbox}
    \begin{example}[Correa and Cristi Algorithm] Input $(\mathcal{D}, \{\mathcal{S}\}_{i = 1}^n)$
        \begin{enumerate}
        \begin{spacing}{0.7}
            \item $\texttt{Initialize $x_i = \varnothing$ for all $i \in N$}$
            \item $\texttt{Sample $v_i' \sim \mathcal{D}_i$ and sample $b_i' \sim \mathcal{S}_i(v_i')$ for all $i \in N$}$
            \item $\texttt{Set $p_j' := \max_{i \in N} b_{i,j}'$ for all $j \in M$}$
            \item $\texttt{When agent $i$ arrives, observe $v_i$ and draw $b_i \sim \mathcal{S}_i(v_i)$}$
            \item $\texttt{Consider the set $S_i = \{j| b_{i,j} > p_j'\}$, and award person $i$ the set $x_i = R_i \cap S_i$}$
            \item \texttt{Repeat step $4$ and $5$ for all $i \in N$}
        \end{spacing}
        \end{enumerate} 
    \end{example}
\end{tcolorbox}

Thus, the algorithm with respect to the valuation profile $\vb{v}$ returns the allocation $\textsf{ALG}(\vb{v}) = (x_1, \dots, x_n)$ and achieves a welfare of $\vb{v}(\textsf{ALG}(\vb{v}))$.

\subsection{Mirror Lemma}
This lemma is proven by Correa and Cristi and provides a lower bound on their algorithm by considering special sets. Before stating the lemma, we first sample $b_i \sim \mathcal{S}_i(v_i)$ for all $i \in N$. Next, we sample two copies that are i.i.d. to the revealed $v_i$, i.e., $v_i', v_i'' \sim \mathcal{D}_i$ and scores $b_i' \sim \mathcal{S}_i(v_i'), b_i'' \sim \mathcal{S}_i(v_i'')$. Define $p_j = \max_{i \in N}b_{i,j}$, $p_j'$ and $p_j''$ analogously for all $j \in M$. Lastly, we define special sets $W_i = \{j | b_{i,j} > \max\{p_j', p_j''\}\}$. Then, for any IRSG $\{\mathcal{S}_i\}_{i = 1}^n$, we have the following lemma:

\begin{lemma}[Mirror Lemma]
\label{Mirror}
    Given a subadditive valuation profile $\vb{v} = (v_i)_{i \in N}$ we have:
    $$\mathbb{E}[\vb{v}(\textsf{ALG}(\vb{v}))] \geq \frac{1}{2} \cdot \sum_{i\in N}\mathbb{E}[v_i(W_i)]$$
\end{lemma}
\begin{proof}
We recall from Correa and Cristi's algorithm that $\mathbb{E}[\vb{v}(\textsf{ALG}(\vb{v}))] = \sum_{i \in N}\mathbb{E}[v_i(R_i \cap S_i)]$, where $R_i$ is the set of remaining items when person $i$ arrives and $S_i = \{j | b_{i,j} > p_j'\}$. We also realize that $R_i$ is independent of the valuation and scores of person $i$, since it is purely dependent on the valuation and scores of people who arrived prior to person $i$. Furthermore, we define an i.i.d. copy of $S_i$ to be $S_i'' = \{j | b_{i,j}'' > p_j'\}$. Using that the joint distribution of the pairs $(v_i, b_i)$ and $(v_i'', b_i'')$ are identical, we write $\mathbb{E}[v_i(S_i \cap R_i)] = \mathbb{E}[v_i''(S_i'' \cap R_i)]$.

Now, we define $R = \{j | p_j' \geq p_j\}$, which is the set of items remaining untouched at the end of the algorithm. It is obvious that for any $i \in N$,  $R \subseteq R_i$. Further, we define $W_i'' = \{j | b_{i,j}'' > \max\{p_j, p_j'\}\}$ and notice that $W_i'' \subseteq S_i''$ since they share the constraint $b_{i,j}'' > p_j'$, except $W_i''$ has an additional constraint of being greater than $p_j$. Lastly, define $\widetilde{R} = \{j | p_j \geq p_j'\}$, which is identically distributed as $R$. This is because swapping $v_i$ with $v_i'$ for all $i \in N$ swaps $p_j$ and $p_j'$ for all $j \in M$. Therefore, since $v_i$ and $v_i'$ are i.i.d., $R$ and $\widetilde{R}$ are i.i.d.. Further, we may see that swapping $v_i$ and $v'_i$ for all $i$ changes nothing about $\max\{p_j,p'_j\}$ for any $j$, and therefore does not change $W''_i$ for any $i$. Thus, we conclude the sets $W''_i \cap R$ and $W''_i \cap \tilde{R}$ are identically distributed. Now we write:
\begin{align*}
    \mathbb{E}[v_i(S_i \cap R_i)]&= \mathbb{E}[v_i''(S_i'' \cap R_i)]\\
    &\geq \mathbb{E}[v_i''(W_i'' \cap R)] 
    = \mathbb{E}[v_i''(W_i'' \cap \widetilde{R})] 
\end{align*}

\noindent The second follows since $W_i'' \subseteq S_i''$, $R \subseteq R_i$ and we assume monotone valuations. The last equality holds since $v_i''(W_i'' \cap R)$ and $v_i''(W_i'' \cap \widetilde{R})$ are identically distributed as shown previously. Now, clearly, we have that $R \cup \widetilde{R} = M$. Thus, using subadditivity, we can write:
$$\mathbb{E}[v_i''(W_i'' \cap \widetilde{R})] +\mathbb{E}[v_i''(W_i'' \cap R)] \geq \mathbb{E}[v_i''(W_i'' \cap M)] = \mathbb{E}[v_i''(W_i'')]$$

\noindent Now we combine our previous results, to write:
\begin{align*}
    \sum_{i \in N}\mathbb{E}[v_i(S_i \cap R_i)] &\geq \sum_{i \in N}\mathbb{E}[v_i''(W_i'' \cap \widetilde{R})]\\
    &= \sum_{i \in N}\frac{1}{2}\left(\mathbb{E}[v_i''(W_i'' \cap \widetilde{R})] + \mathbb{E}[v_i''(W_i'' \cap \widetilde{R})]\right)\\
    &= \sum_{i \in N}\frac{1}{2}\left(\mathbb{E}[v_i''(W_i'' \cap \widetilde{R})] + \mathbb{E}[v_i''(W_i'' \cap R)]\right) \geq \sum_{i \in N}\frac{1}{2}\left(\mathbb{E}[v_i''(W_i'')]\right) = \sum_{i \in N}\frac{1}{2}\left(\mathbb{E}[v_i(W_i)]\right)\\
\end{align*}
\end{proof}

\subsection{Proof of Theorem for Subadditive Valuations}

Throughout this proof, assume that $V_i$ is a finite set for all $i \in N$. We state the main existence theorem:

\begin{theorem}
\label{big guy}
For every $\varepsilon >0$, if all valuations are subadditive, there exists an IRSG such that

$$(6 + \varepsilon) \cdot \mathbb{E}[\vb{v}(\textsf{ALG}(\vb{v}))] \geq \mathbb{E}[\vb{v}(\textsf{OPT}(\vb{v}))] $$
    
\end{theorem}

\noindent Before proving this theorem, we define important terms as well as well-known theorems for clarity.

\begin{definition}
    For any set $S$, we let $\Delta(S)$ be the set of probability distributions on it.
\end{definition}

\begin{definition}[Convexity \cite{hemi}] A set is convex if it includes the line segments joining any two of its points. Or in other words, a set $C$ is convex if whenever $x, y \in C$, the line segment $\{\alpha x + (1 - \alpha)y : \alpha \in [0, 1]\}$ is included in $C$. 
\end{definition}

\begin{theorem}[Heine-Borel Theorem \cite{hemi}] \label{Heine}
A subset of $\mathbb{R}^n$ is compact if and only if it is closed and bounded.
\end{theorem}

\begin{definition}[Upper Hemicontinuity \cite{hemi}] A correspondence $\Phi: X \to Y$ is upper hemicontinuous at point $x$ if for every neighborhood $U$ of $\Phi(x)$, there is a neighborhood $V$ of $x$ such that $z \in V$ implies $\Phi(z) \subset U$
\end{definition}

\begin{theorem}[Kakutani's Fixed Point Theorem \cite{kakutani}]
Let $X$ be a non-empty, compact, and convex subset of $\mathbb{R}^n$, a Euclidean space. We also let $\Phi: X \to 2^X$ be an upper hemicontinuous mapping such that $\forall x \in S$, $\Phi(x)$ is non-empty, closed, and convex. Then there exists a point $x \in X$ such that $x \in \Phi(x)$. That is, $\Phi$ has a fixed point. 
\end{theorem} 

\noindent Now, we prove the main theorem in the remainder of this section. We notate $f(S) = \sum_{j \in S}f_j$ for vector $f \in \mathbb{R}^M$ and set $S \subseteq M$.

\begin{lemma}
\label{helper 0}
    For every $\varepsilon >0$, there exists an IRSG that guarantees that for all $i \in N$:

    $$\mathbb{E} \left[ v_i(W_i)-\sum_{j \in M}b_{i,j} \,\middle\vert\, v_i \right ] \geq \max_{X \subseteq M}\left\{\frac{1}{3} v_i(X) - \mathbb{E}[p'(X)] - \varepsilon \cdot |X |\right\}$$
\end{lemma}

Note here that we only prove the existence of such an IRSG, as opposed to the constructive version in the XOS case \cite{constant}. To prove this, we will use a fixed point approach. Thus, we require a sequence of lemmas before we have the sufficient conditions to use a fixed point theorem and prove the above lemma. We define $v_{max} = \max_{v \in V_i, i \in N} v(M)$. For a given $\varepsilon >0$, let $B_\varepsilon = \{s \cdot \varepsilon \mid s \in \mathbb{N} \text{ and } s \cdot \varepsilon \leq v_{max}\}^M$, a set of $M$ dimensional vectors. Now, below in Lemma \ref{helper 1}, we show an existence result to establish that a specific set is non-empty which we then use in Lemma \ref{helper 2} where we define a mapping to this set.

\begin{lemma}
\label{helper 1}
    For every subadditive and monotone valuation function $v$: $2^M \to \mathbb{R}_{\geq 0}$ with $v(M) \leq v_{max}$ and every $\varepsilon >0$, if $p'$ and $p''$ are random i.i.d. vectors in $\mathbb{R}_{\geq 0}^M$, there exists a vector $f \in B_\varepsilon$ such that for all $X \subseteq M$:
    $$\mathbb{E} \left[ v(\{j \mid f_{j} > \max\{p_j', p_j''\}\}) \right]- f(M) \geq \frac{1}{3} v(X) - \mathbb{E}[p'(X)] - \varepsilon \cdot |X |$$
\end{lemma}
\begin{proof}
Clearly, the right hand side of this equation is a function of solely $X$. Thus if we show there exists a vector $f$ for the tightest lower bound, or the set $X^*$ that maximizes the lower bound, we have proved the statement. Moreover, we show the statement holds for a random vector $\hat{f}$, as we know if it holds in expectation then there exists a deterministic vector in $B_\varepsilon$ that thus satisfies this statement. So, we define $\hat{f}$ as follows where $p'''$ is an i.i.d. copy of $p'$:

\[ \hat{f_j} =\begin{cases} 
      \lfloor p_j'''/\varepsilon \rfloor \cdot \varepsilon + \varepsilon & \text{if } j \in X^* \\
      0 & \text{otherwise} \\
   \end{cases}
\]
By this definition, $\forall j \in X^*$, $\hat{f_j} > p'''$, and it follows due to monotonicity of $v$:
 $$\mathbb{E} \left[ v(\{j \mid \hat{f}_{j} > \max\{p_j', p_j''\}\}) \right] \geq  \mathbb{E} \left[ v(\{j \mid p_{j}''' \geq \max\{p_j', p_j''\}\} \cap X^*) \right]$$
 
\noindent Since $p', p''$ and $p'''$ are i.i.d., we may interchange them inside the expectation:
\begin{align*}
    \mathbb{E} \left[ v(\{j \mid p_{j}''' \geq \max\{p_j', p_j''\}\} \cap X^*) \right] &= \mathbb{E} \left[ v(\{j \mid p_{j}'' \geq \max\{p_j', p_j'''\}\} \cap X^*) \right] \\
    &= \mathbb{E} \left[ v(\{j \mid p_{j}' \geq \max\{p_j''', p_j''\}\} \cap X^*) \right]
\end{align*}
We may add these up:
\begin{align*}
&3 \cdot \mathbb{E} \left[ v(\{j \mid p_{j}''' \geq \max\{p_j', p_j''\}\} \cap X^*) \right] \\&= \mathbb{E} \left[ v(\{j \mid p_{j}''' \geq \max\{p_j', p_j''\}\} \cap X^*) \right] + \mathbb{E} \left[ v(\{j \mid p_{j}'' \geq \max\{p_j', p_j'''\}\} \cap X^*) \right] + \mathbb{E} \left[ v(\{j \mid p_{j}' \geq \max\{p_j''', p_j''\}\} \cap X^*) \right] \\
& \geq \mathbb{E}  \left[v \left( \bigg\{j \mid  p_{j}''' \geq \max\{p_j', p_j''\}\} \text{ or } p_{j}'' \geq \max\{p_j', p_j'''\} \text{ or } 
p_{j}' \geq \max\{p_j''', p_j''\}\bigg\} \cap X^* \right) \right] \\
& = \mathbb{E} \left[v \left(M \cap X^*\right) \right] = \mathbb{E} [v (X^*)]
\end{align*}
where the inequality follows from the subadditivity of $v$ and
$\bigg\{j \mid  p_{j}''' \geq \max\{p_j', p_j''\}\} \text{ or } p_{j}'' \geq \max\{p_j', p_j'''\} \text{ or } 
p_{j}' \geq \max\{p_j''', p_j''\}\bigg\} = M$. Thus, 
$$\mathbb{E} \left[ v(\{j \mid p_{j}''' \geq \max\{p_j', p_j''\}\} \cap X^*) \right] \geq \frac{1}{3} \mathbb{E} [v (X^*)]$$

\noindent Because $\hat{f}_j \leq p_j''' + \varepsilon$ for all $j \in X^*$, we may subtract the expectation of this inequality from the above equation thus:
\begin{equation}
\label{proof end}
    \mathbb{E} \left[ v(\{j \mid \hat{f_{j}} > \max\{p_j', p_j''\}\}) \right]-\hat{f}(M) \geq \frac{1}{3} v(X) - \mathbb{E}[p'(X)] - \varepsilon \cdot |X |
\end{equation}

So if $\hat{f} \in B_\varepsilon$, we are done. However, to make sure $\hat{f} \in B_\varepsilon$, we consider the case that it's not which would imply by definition that there exists a $j \in [m]$ such that $\hat{f_j} > v_{max}$. Thus, we know that $\hat{f}(M) = \sum_{j \in M}\hat{f_j} > v_{max} \geq v(M)$. Furthermore, note that $\{j \mid \hat{f_{j}} > \max\{p_j', p_j''\}\} \subseteq M \Longrightarrow \mathbb{E} \left[ v(\{j \mid \hat{f_{j}} > \max\{p_j', p_j''\}\}) \right] \leq v(M)$. This implies that $\mathbb{E} \left[ v(\{j \mid \hat{f_{j}} > \max\{p_j', p_j''\}\}) \right] -\hat{f_j}(M) < 0$ for such a vector $\hat{f}$. 

Now consider replacing this vector $\hat{f}$ with the zero vector. This makes the LHS of \ref{proof end} exactly $0$, which means the inequality still holds for the zero vector. Furthermore, the zero vector is trivially in $B_{\varepsilon}$. Thus, we can claim that there always exists the desired vector $\hat{f} \in B_{\varepsilon}$.

\end{proof}

\noindent\textbf{Key Mapping:} Now, define the space $\mathcal{L} = \bigtimes_{i \in N} \Delta (B_\varepsilon)^{V_i}$. Breaking this notation down, $\Delta (B_\varepsilon)$ is the set of probability distributions over the set $B_\varepsilon$. So $\Delta (B_\varepsilon)^{V_i}$ is defined for all $v\in V_i$, for a specific $V_i$. That is, we may think of this as the space of RSGs for person $i$ with scores in $B_\varepsilon$. To generalize this for all $i \in N$, we take the Cartesian product over this and we find that $\mathcal{L}$ is the space of IRSGs with scores in $B_\varepsilon$.

Next, we define a set-valued function $\Phi: \mathcal{L} \to 2^\mathcal{L}$. For a given IRSG, $\mathcal{S} = (\mathcal{S}_i)_{i \in N} \in \mathcal{L}$, we sample independently $v_i', v_i'' \sim \mathcal{D}_i$ and $b_i' \sim \mathcal{S}_i(v_i')$, $b_i'' \sim \mathcal{S}_i(v_i'')$ for all $i \in N$. Using these samples, we define $p_j' = \max_{i \in N} b_{i,j}'$ and $p_j'' = \max_{i \in N} b_{i,j}''$. We define $\Phi(\mathcal{S}) \subseteq \mathcal{L}$ to be set of IRSGs $\mathcal{I}$, such that every IRSG $\mathcal{G}= (\mathcal{G}_i)_{i \in N} \in \mathcal{I}$ we have: for all $i \in N$ and $v \in V_i$ if $f \sim \mathcal{G}_i(v)$ that,
\begin{equation}
\label{hi ananya}
    \mathbb{E} \left[ v(\{j \mid f_{j} > \max\{p_j', p_j''\}\}) \right]-f(M) \geq \underset{X \subseteq M}{\max} \left\{\frac{1}{3} v(X) - \mathbb{E}[p'(X)] - \varepsilon \cdot |X |\right\}
\end{equation}

In order to use Kakutani's fixed point theorem, we must prove the sufficient conditions for $\mathcal{L}$ and $\Phi(\mathcal{S})$.

\begin{lemma}
\label{helper 2}
    $\mathcal{L}$ is a nonempty, convex and compact subset of a Euclidean space; the function $\Phi$ is upper hemicontinuous; and for every $\mathcal{S} \in \mathcal{L}$, the set $\Phi(\mathcal{S})$ is non-empty, closed and convex.
\end{lemma}

\begin{proof}
    Before we prove these conditions, we rewrite the space $\mathcal{L}$ into a simpler form. First, note that the size of $B_\varepsilon$ is finite and the size is $|B_\varepsilon| = \lfloor v_{max}/ \varepsilon \rfloor ^{|M|}$ because $s$ iterates through exactly this many values in the set. Now, we know that any element $x \in \mathcal{L}$ is a collection of $\sum_{i \in N} |V_i|$ probability distributions over $B_\varepsilon$. This is because for every valuation function, we define a distribution over $B_\varepsilon$. So, $x$ can be written as a vector of dimension $\lfloor v_{max}/ \varepsilon \rfloor ^{|M|} \cdot \sum_{i \in N} |V_i|$ by appending all of these distributions, where each entry is $\in [0,1]$. Thus, we know that any $x \in \mathcal{L}$ can be written as a subset of $[0,1]^l$ where $l = \lfloor v_{max}/ \varepsilon \rfloor ^{|M|} \cdot \sum_{i \in N} |V_i|$. 
    
    Therefore, we may define $\mathcal{L}$ as the following set where every vector $x$ corresponds to a vector $b$, person $i$, and valuation $v$. Moreover, the sum over scores $b \in B_\varepsilon$ of $x$ is $1$. These entries can be interpreted as the probability that with a fixed $v \in V_i$, for a fixed agent $i$, that the scores are $b \in B_\varepsilon$. 
    
    $$\mathcal{L} = \bigg\{x = (x_{b, i, v})_{b\in B_\varepsilon, i \in N, v \in V_i} \in [0,1]^l \mid \sum _{b\in B_\varepsilon} x_{b, i, v} = 1, \forall i \in N, v\in V_i\bigg\}$$

\noindent Now for $\mathcal{L}$, we prove the conditions:
    \begin{itemize}
        \item \textbf{Nonempty:} We know that $B_{\varepsilon}$ is non-empty as shown earlier. It follows that $\mathcal{L}$ is also non-empty. 
        \item \textbf{Convex:} The set $\mathcal{L}$ is convex due to the following argument. Consider any pair $x, y \in \mathcal{L}$ and $z = \alpha x + (1 - \alpha)y$ for some $\alpha \in [0,1]$, where $z \in \mathbb{R}^l$. First note that for all $i \in N$ and $v \in V_i$, we have that $\sum_{b \in B_{\varepsilon}}x_{b, i, v}=\sum_{b \in B_{\varepsilon}}y_{b, i, v} =1$. Thus, we can write for all $i \in N$ and $v \in V_i$:
        $$\sum_{b \in B_{\varepsilon}}z_{b, i, v} = \alpha\sum_{b \in B_{\varepsilon}}x_{b, i, v} + (1-\alpha)\sum_{b \in B_{\varepsilon}}y_{b, i, v} = 1$$
        Above we substituted and used the fact that $\alpha$ is a scalar. Thus, we know that $z \in\mathcal{L}$, proving that $\mathcal{L}$ is indeed convex.     
        \item \textbf{Compact:} Clearly the set $\mathcal{L}$ is finite and closed thus implying it is compact by Theorem \ref{Heine}.
    \end{itemize}

    \noindent Next we show the following conditions for $\Phi(\mathcal{S})$: 

    \begin{itemize}
        \item \textbf{Nonempty:} $\Phi(\mathcal{S})$ is nonempty as by Lemma \ref{helper 1}, there exists some vector $f$ that satisfies the conditions required be in this set.
        \item \textbf{Convexity and Compactness:} To prove the convex and compact argument, we define the distribution of $p'$ and $p''$. Given some $x \in \mathcal{L}$, where $x = (x_{b, i, v})_{b\in B_\varepsilon, i \in N, v \in V_i}$, the distribution can be denoted by $\pi(x)$ where $\pi(x) \in \Delta(B_\varepsilon)$. That is, $\pi(x)$ is a probability distribution over the set $B_\varepsilon$ of scores. We write this distribution for some $f \in B_\varepsilon$ where $f$ is equal to $p$ or any similar vector (i.e. $p'$, $p''$, etc). We let $\mathbf{v} = (v_1, \dots, v_n)$, a vector of valuation functions for all people $i \in N$ where $v_i \in V_i$. Moreover, we let $\mathbf{b} = (b_1, \dots, b_n)$, a vector of $b_i$'s (which themselves are vectors of scores for each item $j \in M$) for all people $i \in N$. 
        This distribution can thus be defined by conditioning on a fixed set of valuation fucntions for the $n$ people, then multiplying by the probability that the set of valuations is exactly these valuation functions (over all possible valuation functions and sets of scores). Stated differently, this distribution is the probability that $\mathbf{v}$ is exactly some set of valuation functions (and we sum across all possible sets of valuation functions) multiplied by the probability that given such a fixed valuation function that the score $b_i$ for some $j$ is the max (or is equal to $p_j$ over all possible sets of scores $\mathbf{b} \in B_\epsilon ^N$)
        \begin{align*}
        \pi_f(x) &= \sum_{(v_i)_{i \in N} \in \bigtimes_{i \in N} V_i} \Pr[\mathbf{v} = (v_i)_{i \in N}] \cdot \sum_{(b_i)_{i \in N} \in B_\epsilon ^N : \max_{i \in N}b_{i, j} = f_j} \Pr[\mathbf{b} = (b_i)_{i \in N} \mid \mathbf{v} = (v_i)_{i \in N}] \\
        &= \sum_{(v_i)_{i \in N} \in \bigtimes_{i \in N} V_i} \left(\prod_{i \in N} q_{i, v_i}\right)\cdot \sum_{(b_i)_{i \in N} \in B_\epsilon ^N : \max_{i \in N}b_{i, j} = f_j} \left(\prod_{i \in N} x_{b_i, i, v_i}\right)
        \end{align*}

    where we let $q_{i, v_i}$ be the probability $v=v_i$ for some person $i$. The above lines use the law of total probability and independence respectively. Now we may rewrite the constraint of Equation \ref{hi ananya} for the set $\Phi$ using $\pi$, where for $x, y \in \mathcal{L}$, $y \in \Phi(x)$ if and only if for all $i \in N, v \in V_i,\text{ and }X \subseteq M$:
    \begin{align*}
    &\sum_{f \in B_\epsilon}\sum_{p', p'' \in B_\epsilon} y_{f, i, v}\cdot \pi_{p'}(x) \cdot \pi_{p''}(x) \cdot \left(\left( v(\{j \mid f_{j} > \max\{p_j', p_j''\}\}) \right)-f(M)\right) \\ &\geq \left\{\frac{1}{3} v(X) - \sum_{p' \in B_\epsilon}\pi_{p'} (x) \cdot p'(X) - \varepsilon \cdot |X |\right\}
    \end{align*}

    We find the above line by applying the definition of expectation. Since these are clearly linear constraints we know this set is convex. Moreover, there are finitely many of these and thus this set is finite and bounded and therefore compact.
    \end{itemize}

\noindent Now we just need to show $\Phi$ is upper hemicontinuous. If $\Phi$ is continuous in $x$ and $y$ then $\Phi$ is upper hemicontinuous in $x$ and $y$. So we prove continuity in the following lines by showing that for any two convergent sequences $(x^{(k)})_{k \in \mathbb{N}} \in \mathcal{L}$ and $(y^{(k)})_{k \in \mathbb{N}} \in \mathcal{L}$, these converge to values $x^{(\infty)}$ and $y^{(\infty)}$ respectively. We choose these sequences such that $y^{(k)} \in \Phi(x^{(k)}),  \forall k \in \mathbb{N}$. Now looking at the equation above, we realize that $\pi$ is a continuous function in $x$ and $y$. Because this equation holds for $y^{(k)} \in \Phi(x^{(k)})$ for any $k$, it holds for the sequences $(x^{(k)})_{k \in \mathbb{N}}$ and $(y^{(k)})_{k \in \mathbb{N}}$. But since continuous, this equation thus holds for $x^{(\infty)}$ and $y^{(\infty)}$. This implies that $y^{(\infty)}\in \Phi(x^{(\infty)})$ by the if and only if statement above and $\Phi$ is continuous.

\end{proof}

\begin{proof}[Proof of Lemma \ref{helper 0}]
   Lemma \ref{helper 2} gives the conditions required to invoke Kakutani's fixed point theorem, which gives that there is an IRSG, $\mathcal{S}$, such that $\mathcal{S} \in \Phi(\mathcal{S})$ which proves that for $\varepsilon > 0$, there is an IRSG such that:
   \begin{equation}
   \label{lil guy}
       \mathbb{E} \left[ v(\{j \mid f_{j} > \max\{p_j', p_j''\}\}) \right]-f(M) \geq \underset{X \subseteq M}{\max} \left\{\frac{1}{3} v(X) - \mathbb{E}[p'(X)] - \varepsilon \cdot |X |\right\}
   \end{equation}

   \noindent We know that we can rewrite the left-hand side of Equation \ref{lil guy} as follows:
\begin{align*}
     \mathbb{E} \left[ v_i(\{j \mid f_{j} > \max\{p_j', p_j''\}\})\right ]-f(M) &=  \mathbb{E} \left[ v_i(\{j \mid b_{i,j} > \max\{p_j', p_j''\}\})-b_{i}(M) \,\middle\vert\, v_i \right ] \\
    &= \mathbb{E} \left[ v_i(W_i)-b_{i}(M) \,\middle\vert\, v_i \right ]
\end{align*}

\noindent where the above lines follow from replacing $f$ with some vector $b_i$ due to $\mathcal{S}$ being a fixed point. Clearly $b_i$ does depend on the valuation for person $i$ thus we condition on this valuation. Applying the definition of $W_i$ concludes the proof.
\end{proof}

\begin{proof}[Proof of Theorem \ref{big guy}]
    We consider the IRSG $\mathcal{S}$ promised by Lemma \ref{helper 0}, with some $\delta > 0$. Reiterating notation, we have $\textsf{OPT}_i(\vb{v})$ as the optimal allocation to person $i$ under the subadditive profile $\vb{v}$. Applying the Lemma \ref{Mirror} (Mirror Lemma), we get:
    \begin{align*}
        \mathbb{E}[\vb{v}(\textsf{ALG}(\vb{v}))] &\geq \frac{1}{2}\sum_{i \in N}\mathbb{E}[v_i(W_i)]\\
        &= \frac{1}{2}\sum_{i \in N}\mathbb{E}[v_i(W_i) - b_i(M)] + \sum_{i \in N}\mathbb{E}[b_i(M)] \quad\text{by Linearity of Expectation}\\
        &= \frac{1}{2}\sum_{i \in N}\mathbb{E}[\mathbb{E}[v_i(W_i) - b_i(M)|v_i]] + \sum_{i \in N}\mathbb{E}[b_i(M)] \quad \text{by Tower Property}\\
        &\geq \frac{1}{2}\sum_{i \in N}\mathbb{E}\left[\frac{1}{3} v_i(\textsf{OPT}_i(\vb{v})) - \mathbb{E}[p'(\textsf{OPT}_i(\vb{v}))] - \delta \cdot |\textsf{OPT}_i(\vb{v})|\right] + \sum_{i \in N}\mathbb{E}[b_i(M)] \quad \text{by Lemma \ref{helper 0}}
    \end{align*}
    \noindent In the last line, we set $X = \textsf{OPT}_i(\vb{v})$. Note that $\cup_{i \in N}\textsf{OPT}_i(\vb{v}) \subseteq M$. Thus, we can write:
    $$\mathbb{E}[\vb{v}(\textsf{ALG}(\vb{v}))] \geq \frac{1}{6} \cdot \mathbb{E}\left[ \vb{v}(\textsf{OPT}(\vb{v}))\right] - \mathbb{E}[p'(M)] - \delta \cdot |M| + \sum_{i \in N}\mathbb{E}[b_i(M)]$$
    We notice that $\sum_{i \in N}\mathbb{E}[b_i(M)] - \mathbb{E}(p'(M))) \geq 0$, since $\mathbb{E}[p_j']=\mathbb{E}[p_j]=\mathbb{E}[\max_{i \in N}b_{i,j}] \leq \sum_{i\in N}\mathbb{E}[b_{i,j}]$ for all $j \in M$.
    Thus, we can write:
    $$\mathbb{E}[\vb{v}(\textsf{ALG}(\vb{v}))] \geq \frac{1}{6} \cdot \mathbb{E}\left[ \vb{v}(\textsf{OPT}(\vb{v}))\right] - \delta \cdot |M|$$
    Setting $\delta \leq \frac{\varepsilon \cdot \mathbb{E}[ \vb{v}(\textsf{OPT}(\vb{v}))]}{6\cdot(6 + \varepsilon)\cdot |M|}$ proves the theorem.
    
\end{proof}

The proof we presented in this section assumed that supports of the distributions of valuation functions are finite, which was implied by our assumption of $V_i$ being a finite set for all $i \in N$. Furthermore, we assumed that the valuations were uniformly bounded by some constant.

There are techniques to get around both. The first uses a trick where the supports are not finite, but are uniformly bounded in which case they simply discretize them and lose an extra factor of $\varepsilon$. Secondly, we can make the uniformly bounded assumption by truncating the valuations, which can be done since we assumed $\mathbb{E}[\vb{v}(\textsf{OPT}(\vb{v}))] < \infty$. For formal proofs of these facts, we defer to the Section $6$ of \cite{constant}.

\section{Concluding Remarks}

In conclusion, we have provided a comprehensive analysis of prophet inequalities for combinatorial auctions, specifically where the class of valuations are subadditive. We started by defining the problem at heart then giving a brief introduction of significant literature in the field. Next, we thoroughly surveyed the recent paper by Correa and Cristi in which they provide a proof of existence of a constant factor prophet inequality. However, it is important to notice that this proof is \textit{non constructive} and cannot be implemented with prices. Regardless, the paper introduces concepts  such as Random Score Generators, which can be used to ``generate'' prices. Though an IRSG for subadditive values cannot be constructed, its existence is a substantial enough result. In fact, it is possible to construct IRSGs in the case where valuations are fractionally subadditive or XOS; however, they provide a worse guarantee than the current best bound. 

Despite these limitations, the work by Correa and Cristi provides inspiration as well as the possibility of extending these ideas to design a constructive mechanism which achieves a constant approximation while being implementable by prices, such as posted price mechanisms. Such a mechanism would be a significant advancement in the field, and it would have many practical applications in settings where combinatorial auctions are used.

\section{Acknowledgements}

We would like to thank Professors Matthew Weinberg and Huacheng Yu for their kind feedback and support.









\newpage
\bibliographystyle{alpha}
\bibliography{refs} 

\begin{thebibliography}{PDTKL17}

\bibitem[AB06]{hemi}
Charalambos Aliprantis and Kim Border.
\newblock Correspondences. in: Infinite dimensional analysis, 2006.

\bibitem[BR11]{bawal}
Kshipra Bhawalkar and Tim Roughgarden.
\newblock Welfare guarantees for combinatorial auctions with item bidding,
  2011.

\bibitem[CC22]{constant}
José Correa and Andrés Cristi.
\newblock A constant factor prophet inequality for online combinatorial
  auctions, 2022.

\bibitem[Kir11]{kakutani}
A.~Ya. Kiruta.
\newblock Kakutani's theorem, encyclopedia of mathematics, 2011.

\bibitem[KL20]{loglog}
Paul Dütting~Thomas Kesselheim and Brendan Lucier.
\newblock An o(log log m) prophet inequality for subadditive combinatorial
  auctions, 2020.

\bibitem[PDTKL17]{log}
Michael~Feldman Paul Dütting Thomas~Kesselheim and Brendan Lucier.
\newblock Prophet inequalities made easy, stochastic optimization by pricing
  non stochastic inputs, 2017.

\bibitem[Sah23]{dwip}
Dwaipayan Saha.
\newblock O(1) prophet inequality for subadditive combinatorial auctions, 2023.

\end{thebibliography}
\end{spacing}
\end{document}